\documentclass[graybox]{svmult}

\usepackage{mathptmx}       % selects Times Roman as basic font
\usepackage{helvet}         % selects Helvetica as sans-serif font
\usepackage{courier}        % selects Courier as typewriter font
\usepackage{type1cm}        % activate if the above 3 fonts are
\usepackage{makeidx}         % allows index generation
\usepackage{graphicx}        % standard LaTeX graphics tool
\usepackage{multicol}        % used for the two-column index
\usepackage[bottom]{footmisc}% places footnotes at page bottom

\usepackage{amssymb,amsmath,
}

\smartqed

\def\der#1{\frac{\partial}{\partial #1}}  % partial derivative
\def\Ci{\operatorname{Ci}}
\def\Si{\operatorname{Si}}
\def\Ei{\operatorname{Ei}}
\def\E#1{E_{#1}}
\def\e{\mathrm{e}}
\def\i{\mathrm{i}}
\def\d{\mathrm{d}}

\let\set\mathbb
\def\subsec{\subsection}
\newcounter{mmacnt}
\def\restartmma{\setcounter{mmacnt}{0}}
\restartmma
\catcode`|=\active
\def|#1|{\mathrm{#1}}
\catcode`|=12
\newenvironment{mma}{
 \par\smallskip
 \catcode`|=\active
 \parskip=0pt\parindent=0pt % locally
 \small
 \def\In##1\\{%
   \def\linebreak{\hfill\break\null\qquad}%
   \refstepcounter{mmacnt}
   \hangindent=2.5em\hangafter=0
   \leavevmode
   \llap{\tiny\sffamily In[\arabic{mmacnt}]:=\kern.5em}%
   \mathversion{bold}$\displaystyle##1$
   \mathversion{normal}\par
 }%
 \def\Print##1\\{%
   \def\linebreak{\hfill\break}%
   \hangindent=2.5em\hangafter=0
   \leavevmode {\scriptsize##1}\par}%
 \def\Out##1\\{%
   \def\linebreak{$\hfill\break\null\hfill$}%
   \kern\abovedisplayskip\par
   \hangindent=2.5em\hangafter=0
   \leavevmode
   \llap{\tiny\sffamily Out[\arabic{mmacnt}]=\kern.5em}
   $\displaystyle##1$\hfill\null\par
   \kern\belowdisplayskip
 }%
 \def\Warning##1##2\\{%
   \def\linebreak{\hfill\break}%
   \hangindent=2.5em\hangafter=0
   \leavevmode
   {\scriptsize##1 : ##2}\par}%
}{%
 \par\smallskip
}

\newcommand{\myIn}[1]{{\small\sffamily In[#1]}}
\newcommand{\myOut}[1]{{\small\sffamily Out[#1]}}

\def\MLabel#1{{\refstepcounter{mmacnt}\label{#1}}\addtocounter{mmacnt}{-1}}

\makeindex             % used for the subject index
\begin{document}

%PP title change
%\title*{Computer~Proofs of~Some~Identities for~Bessel~Functions  of~Fractional~Order}
\title*{Computer-Assisted Proofs of~Some~Identities for~Bessel~Functions
  of~Fractional~Order}

%PP: title change 
%\titlerunning{Computer-Assisted Proofs of Some Identities for Bessel
\titlerunning{Computer-Assisted Proofs of Some Identities for Bessel
Functions}

\author{
  Stefan Gerhold
  \and
  Manuel Kauers
  \and
  Christoph Koutschan
  \and
  Peter Paule
  \and
  Carsten Schneider
  \and
  Burkhard Zimmermann\\[1cm]
	In memory of Frank W.J. Olver (1924-2013)
}

\authorrunning{S. Gerhold, M. Kauers, C. Koutschan, P. Paule, C. Schneider, B. Zimmermann}

\institute{
Manuel Kauers, Peter Paule, Carsten Schneider, Burkhard Zimmermann \at
Research Institute for Symbolic Computation (RISC), Johannes Kepler University, Linz, Austria\\
\email{FirstName.LastName@risc.jku.at} 
\and
Stefan Gerhold \at
Financial and Actuarial Mathematics, Vienna University of Technology, Vienna, Austria\\
\email{sgerhold@fam.tuwien.ac.at}
\and
Christoph Koutschan \at
Johann Radon Institute for Computational and Applied Mathematics (RICAM),
Austrian Academy of Sciences (\"OAW), Linz, Austria\\
\email{christoph.koutschan@ricam.oeaw.ac.at}
\bigskip\\
This article appeared in:\\
C. Schneider and J. Bl\"umlein (eds.), \emph{Computer Algebra in Quantum Field Theory},\\
Texts \& Monographs in Symbolic Computation, DOI 10.1007/978-3-7091-1616-6\_3,\\
Springer-Verlag Wien 2013
}

\maketitle


\abstract{We employ computer algebra algorithms to prove a collection of
  identities involving Bessel functions with half-integer orders and other
  special functions. These identities appear in the famous Handbook of
  Mathematical Functions, as well as in its successor, the DLMF, but their
  proofs were lost. We use generating functions and symbolic summation
  techniques to produce new proofs for them.}

\section{Introduction}

%PP: Done (April 8, 2013)
%\ToDo{...}

The Digital Library of Mathematical Functions 
%(DLMF) 
\cite{Olver:2012} is the successor of the classical
Handbook of Mathematical Functions \cite{AbSt73} by Abramowitz and Stegun.
Beginning of June 2005 Peter Paule was supposed to meet Frank Olver, the mathematics 
editor of the \cite{Olver:2012}, at the NIST headquarters in Gaithersburg (Maryland, USA). 
On May 18, 2005, Olver sent the following email to Paule:

``The writing of DLMF Chapter~BS\footnote[1]{finally Chapter 10 Bessel Functions} by
Leonard Maximon and myself is now 
largely complete [...] However, a 
problem has arisen in connection
with about a dozen formulas from
Chapter 10 of Abramowitz and Stegun
for which we have not yet tracked
down proofs, and the author of this 
chapter, Henry Antosiewiecz, died about 
a year ago. Since it is the editorial
policy for the DLMF not to state
formulas without indications of 
proofs, I am hoping that you will be 
willing to step into the breach and 
supply verifications by computer 
algebra methods [...] I will fax you 
the formulas later today."

In view of the upcoming trip to NIST, Paule was hoping to be able to 
provide at least some help in this matter. But the arrival of Olver's
fax chilled the enthusiasm quite a bit. Despite containing some identities 
with familiar pattern, the majority of the entries involved Bessel functions
of fractional order or with derivatives applied with respect to the order.

Let us now display the bunch of formulas we are talking about. Here,
$J_\nu(z)$ and $Y_\nu(z)$ denote the Bessel functions of the first and second kind, respectively,
$I_\nu(z)$ and $K_\nu(z)$ the modified Bessel functions,
$j_n(z)$ and $y_n(z)$ the spherical Bessel functions, $P_n(z)$
the Legendre polynomials, and $\Si(z)$ and $\Ci(z)$
the sine and cosine integral, respectively.
Unless otherwise specified, all parameters are arbitrary complex numbers.

\allowdisplaybreaks
\begin{alignat}3
 \tag{10.1.39}
 \frac1z\sin\sqrt{z^2+2zt}&=\sum_{n=0}^\infty\frac{(-t)^n}{n!}y_{n-1}(z)
 &&\hspace*{-1cm}(2|t|<|z|, |\Im(z)|\leq \Re(z))
 \\
 \tag{10.1.40}% corrected
 \frac1z\cos\sqrt{z^2-2zt}&=\sum_{n=0}^\infty\frac{t^n}{n!}j_{n-1}(z)
 && z\neq0
 \\
 \tag{10.1.41}% corrected
 \left[\der\nu j_\nu(z)\right]_{\nu=0}
 &=\frac1z(\Ci(2z)\sin z-\Si(2z)\cos z) && (z\in\set{C}\setminus {]{-\infty},0]})
 \\
 \tag{10.1.42}% corrected
 \left[\der\nu j_\nu(z)\right]_{\nu=-1}
 &=\frac1z(\Ci(2z)\cos z+\Si(2z)\sin z) && (z\in\set{C}\setminus {]{-\infty},0]})
 \\
 \tag{10.1.43}% corrected
 \left[\der\nu y_\nu(z)\right]_{\nu=0}
 &=\frac1z(\Ci(2z)\cos z+[\Si(2z)-\pi]\sin z) && (z\in\set{C}\setminus {]{-\infty},0]})
 \\
 \tag{10.1.44}% corrected % further correction by CK: minus sign on rhs
 \left[\der\nu y_\nu(z)\right]_{\nu=-1}
 &=-\frac1z(\Ci(2z)\sin z-[\Si(2z)-\pi]\cos z) && (z\in\set{C}\setminus {]{-\infty},0]})
 \\
 \tag{10.1.48}
 J_0(z\sin\theta)&=\sum_{n=0}^\infty(4n+1)\frac{(2n)!}{2^{2n}n!^2}j_{2n}(z)P_{2n}(\cos\theta)
 \\
% \tag{10.1.49}
% j_n(2z)&=-n!z^{n+1}\sum_{k=0}^n\frac{2n-2k+1}{k!(2n-k+1)!}j_{n-k}(z)y_{n-k}(z) \qquad (n=0,1,2,\dots)
% \kern-30pt\null
 \tag{10.1.49}
 j_n(2z)&=-n!z^{n+1}\sum_{k=0}^n\frac{2n-2k+1}{k!(2n-k+1)!}j_{n-k}(z)y_{n-k}(z) && \;\;\;(n=0,1,2,\dots)
 \\
 \tag{10.1.52}
 \sum_{n=0}^\infty j_n^2(z)&=\frac{\Si(2z)}{2z}
 \\
 \tag{10.2.30}
 \frac1z\sinh\sqrt{z^2-2\i zt}&=\sum_{n=0}^\infty\frac{(-\i t)^n}{n!}
  \sqrt{\tfrac12\pi/z}I_{-n+\frac12}(z)
  &&\hspace*{-1cm}(2|t|<|z|, |\Im(z)|\leq \Re(z))
 \\
 \tag{10.2.31}
 \frac1z\cosh\sqrt{z^2+2\i zt}&=\sum_{n=0}^\infty\frac{(\i t)^n}{n!}
  \sqrt{\tfrac12\pi/z}I_{n-\frac12}(z) && z\neq0
 \\
 \tag{10.2.32}% corrected. - Another correction by PP (April 8, 2013)
 \left[\der\nu I_\nu(z)\right]_{\nu=1/2}
 &=-\frac{1}{\sqrt{2\pi z}}(\Ei(2z)\e^{-z} + \mathrm{E}_1(2z)\e^z) && (z\in\set{C}\setminus {]{-\infty},0]})
 \\
 \tag{10.2.33}% corrected
 \left[\der\nu I_\nu(z)\right]_{\nu=-1/2}
 &=\frac{1}{\sqrt{2\pi z}}(\Ei(2z)\e^{-z} - \mathrm{E}_1(2z)\e^z) && (z\in\set{C}\setminus {]{-\infty},0]})
 \\
 \tag{10.2.34}% corrected. - Another correction by PP (April 8, 2013)
 \left[\der\nu K_\nu(z)\right]_{\nu=\pm1/2}
 &=\pm\sqrt{\frac{\pi}{2z}} \mathrm{E}_1(2z)\e^z && (z\in\set{C}\setminus {]{-\infty},0]})
% \\
% \tag{10.2.38}
% K_{n+\frac12}(2z)&=
% n!\pi^{-\frac12}z^{n+\frac12}\sum_{k=0}^n
% \frac{(-1)^k(2n-2k+1)}{k!(2n-k+1)!}K_{n-k+\frac12}^2(z)\quad && %(z\in\set{C}\setminus {]{-\infty},0]}, n=0,1,2,\dots)
% \kern-40pt\null
\end{alignat}

The numbering follows that in Abramowitz and Stegun~\cite{AbSt73}, and Olver remarked
on the fax:
``Irene Stegun left a record (without proofs) that (10.1.41)-(10.1.44) have errors:
the factor $\frac{1}{2} \pi$ should not be there, and (10.1.44) also has the wrong sign.
Equations (10.2.32)-(10.2.34) have similar errors. Their correct versions are given 
by [...]''.

In view of these unfamiliar objects and of the approaching trip to NIST, Paule
asked his young collaborators for help. Within two weeks, all identities succumbed 
to the members of the algorithmic combinatorics group of RISC. Moreover, in addition to 
the ~\cite{AbSt73} typos mentioned by Olver, further typos in (10.1.39) and 
(10.2.30) were found. Above we have listed the corrected versions of the formulas,
and when we use the numbering from~\cite{AbSt73},  we refer to the
corrected versions of the formulas here and throughout the paper.

At this place we want to relate the~\cite{AbSt73} numbering 
to the one used in the~\cite{Olver:2012}:
(10.1.39) and (10.1.40) are DLMF entries 10.56.2 and 10.56.1, respectively.
With the help of the rewriting rule DLMF 10.47.3, (10.1.41) and (10.1.42) are
DLMF entries 10.15.6 and 10.15.7, respectively; using
the rule DLMF 10.47.4, (10.1.43) and (10.1.44) are
DLMF entries 10.15.8 and 10.15.9, respectively.
Entry (10.1.48) is DLMF 10.60.10, (10.1.49) is DLMF 10.60.4, and 
(10.1.52) is DLMF 10.60.11.
With the help of DLMF 10.47.8, entry (10.2.30) 
turns into DLMF 10.56.4; and with the help of DLMF 10.46.7,
entry (10.2.31) turns into DLMF 10.56.3.
Formulas (10.2.32) and (10.2.33) are bundled in DLMF entry 10.38.6;
formula (10.2.34) is DLMF 10.38.7.

The goal of our exposition is
to convince the reader that only a very limited amount
of techniques has to be mastered to be able to prove such special function
identities with computer algebra.

Our computer proofs are based on the algorithmic theory of holonomic functions
and sequences, and symbolic summation algorithms.
In the following two sections, we do purely algebraic manipulations;
where necessary, analytical justifications (convergence of series, etc.) are 
given in Section~\ref{se:anal}. 
In general, we rely on the following computer algebra toolbox; underlying
ideas are described in~\cite{KP11,K13}.

\medskip
\textit{Holonomic closure properties.} The packages {\sf\small gfun}~\cite{Salvy:94} (for Maple) and
{\sf\small GeneratingFunctions}~\cite{Mallinger:96} (for Mathematica) are
useful for the manipulation of functions~$f(x)$ that
satisfy linear ordinary differential equations (LODEs) with
polynomial coefficients, as well as for sequences~$f_n$ satisfying
linear recurrence equations (LOREs) with polynomial coefficients.
Such objects are called \emph{holonomic.} It can be shown that
whenever $f(x)$ and $g(x)$ (resp. $f_n$ and~$g_n$) are holonomic,
then so are $f(x)\cdot g(x)$ and $f(x)+g(x)$ (resp. $f_n\cdot g_n$
and $f_n+g_n$). Furthermore, if $f(x)=\sum_{n=0}^\infty f_nx^n$,
then $f(x)$ is holonomic 
%as series 
if and only if $f_n$ is
holonomic as a sequence. The packages \textsf{\small gfun} and {\sf\small GeneratingFunctions}
provide procedures for ``executing closure properties,'' i.e.,
from given differential equations for $f(x)$ and $g(x)$ they can
compute differential equations for $f(x)\cdot g(x)$ and
$f(x)+g(x)$, and likewise for sequences. Also several further
closure properties can be executed in this sense, and there are
procedures for obtaining a recurrence equation for $f_n$ from a
differential equation for its generating function
$f(x)=\sum_{n=0}^\infty f_nx^n$, and vice versa.

\textit{Symbolic summation tools.} The package {\sf\small Zb}~\cite{Paule:95} (for Mathematica) and
the more general and powerful packages {\sf\small Mgfun}~\cite{Chyzak:00} (for Maple),
{\sf\small HolonomicFunctions}~\cite{Koutschan09} and
{\sf\small Sigma}~\cite{Schneider:05,Schneider:07} (both for Mathematica) provide algorithms
to compute for a given definite sum
$S(n,z)=\sum_{k=0}^nf(n,z,k)$ recurrences (in $n$) and/or
differential equations (in $z$). Here the essential assumption is
that the summand $f(n,z,k)$ satisfies certain types of recurrences or
differential equations; see Section~\ref{Sec:SymSum}.

Subsequently, we restrict our exposition to the Mathematica
packages {\sf\small GeneratingFunctions}, {\sf\small Zb}, {\sf\small HolonomicFunctions}, and {\sf\small Sigma}.
In the Appendix, for the reader's convenience we list all formulas from Abramowitz and 
Stegun~\cite{AbSt73} that we apply in our proofs.

As for applications of
differentiating Bessel functions w.r.t.\ order, we mention maximum
likelihood estimation for the generalized hyperbolic distribution, and
calculating moments of the Hartman-Watson distribution. Both
distributions have applications in mathematical
finance~\cite{Pr99,Gerhold:11}. Prause's PhD thesis~\cite{Pr99} in fact
cites formulas (9.6.42)--(9.6.46).

\section{Basic Manipulations of Power Series}\label{se:comp}

Let us now show how to apply these computer algebra tools for
proving identities.
 %Namely, we 
The basic strategy is to determine algorithmically a
differential equation (LODE) or a recurrence (LORE) for both sides
of an identity and check initial conditions.

First we load the package {\sf\small GeneratingFunctions} in the computer
algebra system Mathematica.

\begin{mma}
\In << |GeneratingFunctions.m| \\
\Print GeneratingFunctions Package by Christian Mallinger --
\copyright\ RISC Linz\\
\end{mma}

%%%%%%%%%%%%%%%%%%%%%%%%%%%%%%%%%%%%%%%%%%%%%%%%%%%%%%%%%%%%%%%%%%%%%%%%%%%%%%%%%%
\subsec{LODE and initial conditions for (10.1.39)}
%%%%%%%%%%%%%%%%%%%%%%%%%%%%%%%%%%%%%%%%%%%%%%%%%%%%%%%%%%%%%%%%%%%%%%%%%%%%%%%%%%
 We show that both sides of the equation satisfy the same differential
 equation in~$t$, and then check a suitable number of initial values.

 %In a first step 
%PP
 First we compute a differential equation for the left hand side
 $\tfrac1z\sin\sqrt{z^2+2zt}$.
 We view this function as the composition of $\frac1z\sin(t)$
 with $\sqrt{z^2+2zt}$ and compute a differential equation for it from
 defining equations of the components, by using the command
 AlgebraicCompose. (The last argument specifies the function under
 consideration. This symbol is used both in input and output.)
 \begin{mma}
 \In |AlgebraicCompose|[
   f''[t]==-f[t],
   f[t]^2==z^2+2 z t,
   f[t]] \\
 \Out\label{out:3} z f[t] + f'[t] + (2t+z) f''[t] == 0 \\
 \end{mma}
 \noindent In order to obtain a differential equation for the right hand side,
 we first compute a recurrence equation for the coefficient sequence
 $c_n:=(-1)^n/n!\, y_{n-1}(z)$ from the recurrences of its factors (using
 (10.1.19)). (The coefficient-wise product of power series is called \emph{Hadamard product},
 which explains the name of the command REHadamard.)

 \begin{mma}
 \In |REHadamard|[c[n+1]==-c[n]/(n+1),
   c[n-1]+c[n+1]==(2(n-1)+1)/z\,c[n], c[n]]\\
 \Warning{CanRE::denom} Warning. The input equation will be multiplied
 by its denominator. \\
 \Out z c[n] + (1+n)(1+2n)c[n+1] + (1+n)(2+n) zc[n+2] == 0 \\
 \end{mma}
 
\noindent Then we convert the recurrence equation for $c_n$ into a differential
 equation for its generating function $\sum_{n=0}^\infty c_nt^n$,
 which is the right hand side.

 \begin{mma}
 \In |RE2DE|[\%, c[n], f[t]] \\
 \Out\label{out:4} z f[t] + f'[t] + (2t+z) f''[t] == 0 \\
 \end{mma}
\noindent This agrees with output~\ref{out:3}.
 To complete the proof, we need to check two initial values.
 \begin{mma}
 \In |Series|[1/z\,|Sin|[\sqrt{z^2+2zt}],\{t,0,1\}]\\
 \Out \frac{|Sin|[\sqrt{z^2}]}z
    + \frac{\sqrt{z^2}|Cos|[\sqrt{z^2}]}{z^2} t
    + %\Bigl(-\frac{\sqrt{z^2}|Cos|[\sqrt{z^2}]}{2z^3}
       %-\frac{|Sin|[\sqrt{z^2}]}{2z}\Bigr)t^2+
    O[t]^2 \\
% \In \%\ /.\ |Sqrt|[z^2] \to z \\
% \Out \frac{|Sin|[z]}z+\frac{|Cos|[z]}z t
%   +
% \Bigl(-\frac{|Cos|[z]}{2z^2}-\frac{|Sin|[z]}{2z}\Bigr)t^2+O[t]^3\\
 \end{mma}
\noindent By (10.1.12) and (10.1.19), this agrees with the initial
 values of the right hand side for $z\in\set R_{\geq0}$. The extension to
 complex~$z$ will be discussed in Section~\ref{se:anal}.
%\end{proof}

Alternatively, we could have derived a differential equation only for the
right hand side and then check with Mathematica that the left hand side
satisfies this equation:
\begin{mma}
 \In |Out|[\ref{out:4}]\ /.\ f \to (1/z\,|Sin|[\sqrt{z^2+2 z \#}]\&)\\
 \Out |True| \\
\end{mma}

The proofs for (10.1.40), (10.2.30), and (10.2.31) follow the same scheme
as the proof above. Both variants of the proof work in each case.\\
In summary, the most systematic way is to compute a differential
equation for the difference of left hand side and right hand side, and then
check that an appropriate number of initial values are zero.

%{\bf LODE and initial conditions for (10.1.40).}
 %The proof strategy is the same as for identity (10.1.39) above.
 %After typing the commands
 %\begin{mma}
 %\In |AlgebraicCompose|[f''[t]==-f[t],f[t]^2==z^2-2zt,f[t]]\\
 %\Out -zf[t]+f'[t]+(2t-z)f''[t]==0\\
 %\In |REHadamard|[c[n+1]==c[n]/(n+1),c[n-1]+c[n+1]==(2(n-1)+1)/z\,
 %c[n], c[n]]\\
 %\Out z c[n] - (1+n)(1+2n)c[n+1]+(1+n)(2+n)z c[n+2]==0\\
 %\In |RE2DE|[\%, c[n],f[t]]\\
 %\Out zf[t]-f'[t]-(2t-z)f''[t]==0\\
 %\end{mma}
 %we are left with checking two initial values.
 %Comparison of (10.1.11) with
 %\begin{mma}
 %\In
 %|Series|[1/z\,|Cos|[|Sqrt|[z^2-2zt]],\{t,0,2\}]\ /.\ |Sqrt|[z^2]\to
 %z\\
 %\Out
 %\frac{|Cos|[z]}z+\frac{|Sin|[z]}zt+
 % \Bigl(-\frac{|Cos|[z]}{2z}+\frac{|Sin|[z]}{2z^2}\Bigr)t^2+O[t]^3
 %\\
 %\end{mma}
 %completes the proof for $z\in\set R_{\geq0}$.

%{\bf LODE and initial conditions for (10.2.30) and (10.2.31).}
%  These identities are analogous to (10.1.39) and (10.1.40) above, and
%  they can be proved in precisely the same way.
%  We only mention that left and right hand side of (10.2.30) are
%  annihilated by $\i z  + \d/\d z + (2t+\i z)\d^2/\d z^2$
%  and that left and right hand side of (10.2.31) are annihilated by
%  $-\i z + \d /\d z + (2t-\i z)\d^2/\d z^2$.
%  Initial values can be checked using (10.1.1) in connection with
%  (10.1.11) and (10.1.19).
%  The details are left to the reader.

%%%%%%%%%%%%%%%%%%%%%%%%%%%%%%%%%%%%%%%%%%%%%%%%%%%%%%%%%%%%%%%%%%%%%%%%%%%%%%%%%%
\subsec{Proof of (10.1.41)}
%%%%%%%%%%%%%%%%%%%%%%%%%%%%%%%%%%%%%%%%%%%%%%%%%%%%%%%%%%%%%%%%%%%%%%%%%%%%%%%%%%
This time we will not derive an LODE, but instead a recurrence relation
for the Taylor coefficients of the difference of the left and the right hand side.
The term $\log(z/2)$ that occurs in the pertinent expansion~(9.1.64)
is not analytic at $z=0$, hence we
first treat that one ``by hand.'' (Working with Taylor series at $z=1$, say,
promises not much but additional complications.)
This will leave us with a rather complicated expression for a
holonomic formal power series, for which we have to prove that it is
zero. At this point, we will employ the {\sf\small GeneratingFunctions} package
for computing a recurrence equation for the coefficient sequence of
that series. Upon checking a suitable number of initial values,
zero equivalence is then established.

One might think that we would not even have to compute the
recurrences, since it is known a priori that the sum of two sequences
satisfying recurrences of order $r_1$ and $r_2$, respectively,
satisfies a recurrence of order at most $r_1+r_2$.
The same holds for products, with $r_1 r_2$ instead of $r_1+r_2$.
The catch is that the leading coefficient of the combined recurrence might have
roots in the positive integers. It is clear that in order to give an inductive proof
there must not be an integer root beyond the places where we check
initial values.

\begin{proposition}
 Identity (10.1.41) holds for $z\in\set{C}\setminus\set{R}_{\leq 0}$.
\end{proposition}
\begin{proof}
 First we consider the left hand side.
 Using (10.1.1) and (9.1.64) from the Appendix, we get
 \[
  \der\nu j_\nu(z)
  =j_\nu(z)\log\frac z2-\frac12\sqrt\pi\sum_{n=0}^\infty(-1)^n\frac{\psi(\nu+n+\tfrac32)}{\Gamma(\nu+n+\tfrac32)}\frac{(\tfrac14z^2)^n}{n!},
 \]
 where $\Gamma(x)$ and $\psi(x)=\frac{\frac{d}{dx}\Gamma(x)}{\Gamma(x)}$ denote the Gamma and  digamma function, respectively. Hence, with (10.1.11),
 \[
  \left[\der\nu j_\nu(z)\right]_{\nu=0}
  =\frac{\sin z }z\log\frac z2-\frac12\sqrt\pi\sum_{n=0}^\infty(-1)^n\frac{\psi(n+\tfrac32)}{\Gamma(n+\tfrac32)}\frac{(\tfrac14z^2)^n}{n!}.
 \]
 For the right hand side, we need (5.2.14), (5.2.16),
 and the Taylor expansions of $\sin z$ and~$\cos z$.
 We have to show that
 \allowdisplaybreaks
 \begin{alignat*}1
   &
   \left[\der\nu j_\nu(z)\right]_{\nu=0}
   - \bigl(\Ci(2z)\sin z-\Si(2z)\cos z\bigr)/z\notag\\
   ={}&\frac{\sin
   z}z\log\frac z2-\frac12\sqrt\pi\sum_{n=0}^\infty
    \frac{(-1)^n\psi(n+\tfrac32)}{\Gamma(n+\tfrac32)}\frac{(\tfrac14z^2)^n}{n!}
    +\frac{\Si(2z)\cos z}z
    \notag\\
    &\quad{}
   - \frac{\sin z}z\Bigl(\gamma+\log
   (2z)+\sum_{n=1}^\infty\frac{(-1)^n(2z)^{2n}}{2n(2n)!}\Bigr)\notag\\
   ={}&-\frac12\sqrt\pi\sum_{n=0}^\infty\frac{(-1/4)^n\psi(n+\tfrac32)z^{2n}}{\Gamma(n+\tfrac32)n!}
     +2\sum_{n=0}^\infty\frac{(-4)^nz^{2n}}{(2n+1)(2n+1)!}\sum_{n=0}^\infty\frac{(-1)^nz^{2n}}{(2n)!}
  \smash{\quad\raisebox{-6mm}{$\left.\rule{0pt}{12mm}\right\}\ (\ast)$}\kern-1em}\\
   &\quad{}-\bigl(\gamma+2\log2\bigr)\sum_{n=0}^\infty\frac{(-1)^nz^{2n}}{(2n+1)!}
   +4z^2\sum_{n=0}^\infty\frac{(-4)^nz^{2n}}{2(n+1)(2(n+1))!}\sum_{n=0}^\infty\frac{(-1)^nz^{2n}}{(2n+1)!}
 \end{alignat*}
 is identically zero, i.e., $c_n=0$ for all~$n\geq0$, where $c_n$ is
 defined as $(\ast)=\sum_{n=0}^\infty c_nz^{2n}$.

 To this end, we compute step by step a recurrence equation for $c_n$
 from the various coefficient sequences appearing in~$(\ast)$.
 We suppress some of the output, in order to save space.
Recurrences for most of the inner coefficient sequences are easy
 to obtain. For instance, for
 \begin{mma}
 \In f[n\_]:=\frac{(-4)^n}{(2n+1)(2n+1)!}\\
 \end{mma}
\noindent we have
 \begin{mma}
 \In |FullSimplify|[f[n+1]/f[n]]\\
 \Out \frac{-2(2n+1)}{(n+1)(2n+3)^2}\\
 \end{mma}
\noindent and hence the recurrence
 $f_{n+1}=\frac{-2(2n+1)}{(n+1)(2n+3)^2}f_n$.
 Only the series involving $\psi(n+\frac32)$ requires a bit more work.
 Here, we use the package {\sf\small GeneratingFunctions} to obtain a recurrence from the recurrence
 (6.3.5) for $\psi(n+\frac32)$ and the first order recurrence of $(-1/4)^n/\Gamma(n+\tfrac32)n!$.
 \begin{mma}
 \In |recSum| = |REHadamard|[f[n + 1] == f[n] + \frac1{n + 3/2},\linebreak
    f[n + 1] == \frac{-1}{2(2n + 3)(n + 1)}f[n], f[n]]; \\
 \end{mma}
\noindent Next, we compute recurrence equations for the coefficient sequence of
 the two series products in~$(\ast)$.
 \begin{mma}
 \In |recSiCos| =
    |RECauchy|[f[n + 1] == \frac{-2(2n+1)}{(n+1)(2n+3)^2}f[n],\linebreak
       f[n + 1] == \frac{-1}{2(2n+1)(n+1)}f[n], f[n]];\\
 \In |recCiSin| =
  |RECauchy|[f[n + 1] == \frac{-2(n + 1)}{(n + 2)^2(2n + 3)}f[n],\linebreak
    f[n + 1] == \frac{-1}{2(n+1)(2n+3)}f[n], f[n]];\\
 \end{mma}
\noindent The latter recurrence has to be shifted by~1, owing to the
 factor~$z^2$.
 \begin{mma}
 \In |recCiSin| = |recCiSin|\ /.\ f[n\_]\to f[n+1]\ /.\ n\to n-1;\\
 \end{mma}
\noindent The recurrences collected so far can now be combined to a recurrence
 for~$c_n$.
 \begin{mma}
 \In |rec1|=|REPlus|[|recSiCos|,|recSum|,f[n]];\\
 \In
 |rec2|=|REPlus|[|recCiSin|,f[n+1]==\frac{-1}{2(n+1)(2n+3)}f[n],f[n]];\\
 \In |rec|=|REPlus|[|rec1|,|rec2|,f[n]]\\
 \Out 5184( 227 + 60n ) f[n] + \cdots\linebreak
   \hbox to.66\hsize{\dotfill}\linebreak
   \dots
  +
%    36( 1649187 + 1991356n + 817664n^2 + 107520n^3 ) f(1 + n) +
%    3( 829294407 + 1474069500n + 1074937664n^2 + 397998336n^3 + 74188800n^4 +
%       5529600n^5 ) f(2 + n) + 2
%     ( 39142299429 + 82269829612n + 76407843880n^2 + 40415435808n^3 + 13061980224n^4 +
%       2563312896n^5 + 281346048n^6 + 13271040n^7 ) f(3 + n) +
%    4( 4 + n ) ( 9 + 2n )
%     ( 9595442293 + 15719514406n + 11513620596n^2 + 4967205360n^3 + 1372516992n^4 +
%       241426944n^5 + 24643584n^6 + 1105920n^7 ) f(4 + n) +
%    40( 4 + n ) ( 5 + n ) ( 9 + 2n )
%     ( 11 + 2n ) ( 146925155 + 181167172n + 89281632n^2 + 22009824n^3 +
%       2707296n^4 + 132480n^5 ) f(5 + n) +
   7600( 4 + n ) ( 5 + n ) {( 6 + n ) }^2
    ( 9 + 2n ) ( 11 + 2n ) {( 13 + 2n ) }^2
    ( 167 + 60n ) f[n+6] = 0\\
 \end{mma}
\noindent The precise shape of the recurrence is irrelevant, it only matters
 that it has order~$6$ and that the coefficient of $f[n+6]$ (i.e., of $c_{n+6}$) does not
 have roots at nonnegative integers.
 As this is the case, we can complete the proof by checking 
 that the coefficients of $z^0,\dots,z^{10}$ in~($\ast$) vanish,
 which can of course be done with Mathematica.
 %
 %These are easily computed by truncating the series in ($\ast$). We obtain
 %\begin{mma}
 %\In
 % -\frac12\,|Sqrt|[|Pi|]\,
 %   |Sum|[\frac{(-1/4)^n|PolyGamma|[n + 3/2]z^{2n}}{|Gamma|[n +
 %     3/2]n!}, \{n, 0, 6\}]
 % \linebreak
 % + 2\,|Sum|[\frac{(-4)^nz^{2n}}{(2n + 1)(2n + 1)!}, \{n, 0, 6\}]
 %    \,|Sum|[\frac{(-1)^n z^{2n}}{(2n)!}, \{n, 0, 6\}]
 % \linebreak
 % -(|EulerGamma| + 2\,|Log|[2])\,|Sum|[\frac{(-1)^n z^{2n}}{(2n + 1)!},
 % \{n, 0, 6\}]
 % \linebreak
 % + 4 z^2\,|Sum|[\frac{(-4)^n z^{2n}}{2(n + 1)(2(n + 1))!}, \{n, 0, 6\}]
 %     \,|Sum|[\frac{(-1)^n z^{2n}}{(2n + 1)!}, \{n, 0, 6\}]]\\
 %\Out
%\tfrac{9957151}{5890419190656000} z^{14}
%- \tfrac{1568267}{2356167676262400} z^{16}
%+ \tfrac{12265207}{265068863579520000} z^{18}\linebreak
%- \tfrac{149085319}{111328922703398400000} z^{20}
%+ \tfrac{1109921}{51025756239057600000} z^{22}\linebreak
%- \tfrac{1333}{5566446135169920000} z^{24}
%+ \tfrac{1}{463870511264160000}z^{26}
% \\
% \end{mma}
% The coefficients of $z^n$ for $n=0,\dots,12$ in this polynomial agree
% with the coefficients of the respective terms in $\sum_{n=0}^\infty
% c_nz^{2n}$.
% As they are zero, the proof is complete.

Alternatively, a similar proof can be obtained more conveniently using
the package
\begin{mma}
\In << |HolonomicFunctions.m| \\
\Print HolonomicFunctions package by Christoph Koutschan, RISC-Linz, Version 1.6 (12.04.2012) \\
\end{mma}

\noindent One of the main features of this package is the {\rm Annihilator} command; it
analyzes the structure of a given expression and executes the necessary
closure properties automatically, in order to compute a system of differential
equations and/or recurrences for the expression. We apply it to~$(\ast)$:
\begin{mma}
\In |Annihilator|\bigg[
  - \frac{|Sin|[z]}{z} \left(|EulerGamma| + 2\,|Log|[2] + |Sum|\left[\frac{(-1)^n \, (2z)^{2n}}{2n \, (2n)!}, \{n, 1, \infty\}\right]\right)
  \linebreak
  - \frac{\sqrt{\pi}}{2} |Sum|\left[ \frac{(-1/4)^n \, z^{2n} \, |PolyGamma|[0,n+3/2]}{n! \, |Gamma|[n+3/2]}, \{n, 0, \infty\}\right]
  + \frac{|Cos|[z]}{z}|SinIntegral|[2z],
  \linebreak |Der|[z] \bigg]\\
\Out \big\{
  (48z^5+95z^3) D_{\!z}^8 +
  (864z^4+1900z^2) D_{\!z}^7 +
  (576z^5+5436z^3+10830z) D_{\!z}^6 + {} \linebreak
  (7968z^4+23684z^2+17100) D_{\!z}^5 +
  (1440z^5+32442z^3+77002z) D_{\!z}^4 + {} \linebreak
  (13344z^4+59332z^2+83448) D_{\!z}^3 +
  (1344z^5+33596z^3+82858z) D_{\!z}^2 + {} \linebreak
  (6240z^4+31404z^2+46892) D_{\!z} +
  (432z^5+6495z^3+15150z)
  \big\}\\
\end{mma}
\noindent Since the {\sf\small HolonomicFunctions} package uses operator notation, the second
argument indicates that a differential equation w.r.t.~$z$ is desired; instead
of an equation the corresponding operator is returned with
$D_{\!z}=\mathrm{d}/\mathrm{d}z$. As before, the proof is completed by
checking a few initial values (see also Section~\ref{se:anal}).\qed
\end{proof}

\section{Symbolic Summation Tools}\label{Sec:SymSum}

It is not always the case that recurrences for the power series coefficients
can be obtained by the package {\sf\small GeneratingFunctions}.
Sometimes combinatorial identities such as the following one are needed.
Its proof gives occasion to introduce the Mathematica package {\sf\small Zb},
an implementation of Zeilberger's algorithm for hypergeometric summation~\cite{Zeilberger91}.

\begin{lemma}\label{le:by zb}
  For $k\in\mathbb{Z}_{\geq 0}$ we have
  \[
    \sum_{j=1}^k \frac{(-2)^j}{j}\binom{k}{j} =
    \begin{cases}
       \mathrm{H}_{n+1} - 2\mathrm{H}_{2n+2} & k=2n+1\quad\text{is odd} \\
       \mathrm{H}_n - 2\mathrm{H}_{2n} & k=2n\quad\text{is even,}
    \end{cases}
  \]
  where $\mathrm{H}_n:=\sum_{k=1}^n\frac1k$ denotes the harmonic numbers.
\end{lemma}

It can be a chore to locate such identities in the literature. The
closest match that the authors found is the similar identity
$\sum_{j=1}^k (-1)^{j+1}j^{-1} \binom{k}{j} = \mathrm{H}_k$~\cite[p.\ 281]{GrKnPa94}.
Thus, an automatic identity checker like the one we
describe now is helpful. We note in passing that we can not only
verify such identities, but even compute the right hand side from
the left hand side~\cite{Schneider:05}.

\begin{proof}[of Lemma~\ref{le:by zb}]
  We denote the sum on the
  left hand side by~$a_k$.
  Using the Mathematica package
  
  \begin{mma}
  \In << |Zb.m|\\
  \Print Fast Zeilberger Package by Peter Paule and Markus Schorn (enhanced by Axel Riese)
     -- \copyright\ RISC Linz\\
  \end{mma}
\noindent we find
  \begin{mma}     
  \In |Zb|[(-2)^j/j\,|Binomial|[2n+1,j],\{j,1,2n+1\},n]\\
  \Print If `1 + 2 n' is a natural number, then:\\
  \Out\label{out:10} \{(n+1)(2n+3)|SUM|[n]-(4n^2+14n+13)|SUM|[n+1]\linebreak
     {}+(n+2)(2n+5)|SUM|[n+2]==-2\}\\
  \In |Zb|[(-2)^j/j\,|Binomial|[2n,j],\{j,1,2n\},n]\\
  \Print If `2n' is a natural number, then:\\
  \Out\label{out:11} \{(n+1)(2n+1)|SUM|[n]-(4n^2+10n+7)|SUM|[n+1]\linebreak
     {}+(n+2)(2n+3)|SUM|[n+2]==-2\}\\
  \end{mma}
\noindent  hence the sequence $a_k$ satisfies the recurrences
  \[
    (n+1)(2n+3)a_{2n+1} - (4n^2+14n+13)a_{2n+3} + (n+2)(2n+5)a_{2n+5} = -2
  \]
  and
  \[
    (n+1)(2n+1)a_{2n} - (4n^2+10n+7)a_{2n+2} + (2n+3)(n+2)a_{2n+4} = -2.
  \]
  The right hand side satisfies these recurrences, too:
  \begin{mma}
  \In |Out|[\ref{out:10}]
    \ /.\ |SUM|[n\_] \to |HarmonicNumber|[n+1] - 2|HarmonicNumber|[2n+2]\linebreak
    \ //\ |ReleaseHold|\ //\ |FullSimplify|\\
  \Out \{|True|\}\\
  \In
   |Out|[\ref{out:11}]/.\ |SUM|[n\_] \to |HarmonicNumber|[n] - 2|HarmonicNumber|[2n]\linebreak
   \ //\ |ReleaseHold|\ //\ |FullSimplify|\\
  \Out \{|True|\} \\
  \end{mma}
\noindent  Hence the desired result follows by checking the initial conditions $k=0,1,2,3$.\qed
%
%  With Zeilberger's algorithm we can also compute the somewhat simpler recurrence
%  \[
%    (k+1)a_k + a_{k+1} - (k+2)a_{k+2} = 2
%  \]
%  satisfied by $a_k$,  but then plugging in the right hand side is a bit more cumbersome.
\end{proof}

\begin{proposition}\label{pr:10.2.32}
  Identities (10.2.32) and (10.2.33) follow from Lemma~\ref{le:by zb}.
  They hold for $z\in\set{C}\setminus\set{R}_{\leq 0}$.
\end{proposition}
\begin{proof}
  We do Taylor series expansion on both sides of~(10.2.32), and then compare coefficients.
  Using the expansions~(5.1.10) and (5.1.11)
  %\[
  %  \Ei (z) = \gamma + \log z + \sum_{n=1}^\infty \frac{z^n}{n n!}
  %  \qquad \text{and} \qquad
  %  \E1(z) = -\gamma - \log z - \sum_{n=1}^\infty \frac{(-z)^n}{n n!}
  %\]
  and computing Cauchy products, we find that the right hand side of~(10.2.32) equals
  \begin{equation}\label{eq:rhs 10.2.32}
    \sqrt{\frac{2}{\pi z}}\left(\left( \log z + \log 2 + \gamma \right)\sinh z +
     \sum_{n=0}^\infty \frac{a_{2n+1}}{(2n+1)!}z^{2n+1} \right),
  \end{equation}
  where $a_k$ is the sum from Lemma~\ref{le:by zb}.
  The expansion of the left hand side of~(10.2.32) can be done with~(9.6.10) and~(9.6.42).
  Since
  \[
    \frac{(z^2/4)^n}{\Gamma(n+\tfrac32)n!} = \frac{2z^{2n}}{\sqrt{\pi}(2n+1)!}
  \]
  and
  \[
    \psi(n+\tfrac32) = -\gamma - 2\log 2 + 2\mathrm{H}_{2n+2} - \mathrm{H}_{n+1},
  \]
  the left hand side of~(10.2.32) turns out to be
  \begin{equation}\label{eq:lhs 10.2.32}
    \sqrt{\frac{2}{\pi z}}\left(\left( \log z + \log 2 + \gamma \right)\sinh z +
     \sum_{n=0}^\infty \left( \mathrm{H}_{n+1} - 2\mathrm{H}_{2n+2} \right) \frac{z^{2n+1}}{(2n+1)!} \right).
  \end{equation}
  Lemma~\ref{le:by zb} completes the coefficient comparison.
%
%  Both sides of~(10.2.32) are clearly analytic for $z\in\set{C}\setminus\set{R}_{\leq 0}$,
%  hence the identity is proved (cf.\ Section~\ref{se:anal}).

  Identity~(10.2.33)
  can be proved analogously; replace sinh by cosh and $2n+1$ by $2n$ in~\eqref{eq:rhs 10.2.32},
  and sinh by cosh and the summand by $(\mathrm{H}_n - 2\mathrm{H}_{2n})z^{2n}/(2n)!$ in~\eqref{eq:lhs 10.2.32}.\qed
\end{proof}

We proceed to prove the identities (10.1.48), (10.1.49),
and (10.1.52) by the same strategy as above: compute LODEs
or LOREs for both sides, and check initial values. Since in these
identities definite sums occur for which one cannot derive LOREs
or LODEs by using holonomic closure properties, symbolic summation algorithms enter the game. For hypergeometric sums, like in Lemma~\ref{le:by zb}, the package {\sf\small Zb} is the perfect choice. Since in the following identities the occurring sums do not have hypergeometric summands, we use more general summation methods~\cite{Schneider:05} and~\cite{Koutschan09} that are available in the packages \textsf{\small Sigma} and \textsf{\small HolonomicFunctions}, respectively.

In general, the sums under consideration are of the
form
\begin{equation}\label{SigmaInput}
S(n,z)=\sum_{k=0}^{\infty}h(n,k)f(n,z,k)
\end{equation}
with integer parameter $n$ and complex parameter $z$ where $h$ and
$f$ have the following properties: $h(n,k)$ is a hypergeometric
term in $n$ and $k$, i.e., $h(n+1,k)/h(n,k)$ and $h(n,k+1)/h(n,k)$
are rational functions in $n$ and $k$. Furthermore,
$f(n,z,k)$ satisfies a recurrence relation of the form
\begin{multline}\label{recK}
f(n,z,k+d)=\alpha_0(n,z,k)f(n,z,k)\\
+\alpha_1(n,z,k)f(n,z,k+1)+\dots+\alpha_{d-1}(n,z,k)f(n,z,k+d-1),
\end{multline}
and either a recurrence relation
\begin{multline}\label{recN}
f(n+1,z,k)=\beta_0(n,z,k)f(n,z,k)\\
+\beta_1(n,z,k)f(n,z,k+1)+\dots+\beta_{d-1}(n,z,k)f(n,z,k+d-1)
\end{multline}
or a differential equation
\begin{multline}\label{recZ}
\frac{\d}{ \d z}f(n,z,k)=\beta_0(n,z,k)f(n,z,k)\\
+\beta_1(n,z,k)f(n,z,k+1)+\dots+\beta_{d-1}(n,z,k)f(n,z,k+d-1),
\end{multline}
where the $\alpha_i,\beta_i$ are rational functions in $k$, $n$, and $z$.
From recurrences of the forms~\eqref{recK} and \eqref{recN} we will derive a recurrence relation in $n$
for $S(n,z)$. If, on the other hand, we have~\eqref{recZ} instead of~\eqref{recN},
we will compute a differential equation for $S(n,z)$ in~$z$.

We note that the \textsf{\small HolonomicFunctions} package allows more flexible recurrence/ differential systems as input specifying the shift/differential behavior of the summand accordingly. However, the input description given above gives rise to rather efficient algorithms implemented in the \textsf{\small Sigma} package to calculate LOREs and LODEs for $S(n,z)$.
 
%%%%%%%%%%%%%%%%%%%%%%%%%%%%%%%%%%%%%%%%%%%%%%%%%%%%%%%%%%%%%%%%%%%%%%%%%%%%%%%%%%
\subsec{LORE and initial conditions for (10.1.49)}
%%%%%%%%%%%%%%%%%%%%%%%%%%%%%%%%%%%%%%%%%%%%%%%%%%%%%%%%%%%%%%%%%%%%%%%%%%%%%%%%%%
We compute a LORE for the right hand side
\begin{align*}
S(n):=&\sum_{k=0}^n-n!z^{n+1}\frac{2n-2k+1}{k!(2n-k+1)!}j_{n-k}(z)y_{n-k}(z)\\
=&\sum_{k=0}^n\frac{-n!z^{n+1}(2k+1)}{(n-k)!(n+k+1)!}j_{k}(z)y_{k}(z)
\end{align*}
using

\begin{mma}
\In << |Sigma.m| \\
\Print Sigma - A summation package by Carsten Schneider
\copyright\ RISC-Linz\\
\end{mma}

\noindent First we insert the sum in the form~\eqref{SigmaInput}
with recurrences of the type~\eqref{recK} and~\eqref{recN}. Note
that $h(n,k)=\frac{-n!z^{n+1}(2k+1)}{(n-k)!(n+k+1)!}$ is
hypergeometric in $n$ and $k$. Moreover, by (10.1.19) the
spherical Bessel functions of the first kind $j(k):=j_k(z)$ (we
suppress the parameter $z$ in our Mathematica session) fulfill the
recurrence
\begin{mma}\MLabel{MMA:recJ}
\In |recJ|=z j[k]-(2k+3)j[k+1]+z j[k+2]==0;\\
\end{mma}
\noindent Since the same recurrence holds for $y_k(z)$,
see~(10.1.19), we obtain with
\begin{mma}\MLabel{MMA:recJY}
\In |recJY|=|REHadamard|[|recJ|,|recJ|,|j|[k]]/.\{|j|\rightarrow |f|\};\\
\Out (-2 k-5) z^2 |f|[k]+(2 k+3) (4
   k^2+16 k-z^2+15)f[k+1]\linebreak
   -(2 k+5) (4 k^2+16
   k-z^2+15) f[k+2]+(2 k+3)z^2
   f[k+3]=0\\
\end{mma}
\noindent a recurrence in the form~\eqref{recK} for
$f(k):=j_k(z)y_k(z)$. Since $f(k)$ is free of $n$, we choose
$f[n+1,k]==f[k]$ for the required recurrence of the
form~\eqref{recN}. Given these recurrences we are ready to compute
a recurrence for our sum
\begin{mma}
\In |mySum|=\sum_{k=0}^n\frac{-n!z^{n+1}(2k+1)}{(n-k)!(n+k+1)!}|f|[k];\\
\end{mma}
\noindent by using the {\sf\small Sigma}-function
\begin{mma}\MLabel{MMA:RecForjy}
\In |GenerateRE|[|mySum|,n,\{|recJY|,f[k]\},f[n+1,k]==f[k]]\\
\Out 2z|SUM|[n]-(2n+3)|SUM|[n+1]+2z|SUM|[n+2]==0\\
\end{mma}
\noindent Note that
$S(n)=\sum_{k=0}^nh(n,k)f(k)(=\textrm{mySum}=\textrm{SUM}[n])$.
Since besides $S(n)$ also $j_n(2z)$ fulfills the computed
recurrence and since $S(n)=j_n(2z)$ for $n=0,1$, we have
$S(n)=j_n(2z)$ for all $n\geq0$.

\medskip

\noindent \textit{A correctness proof.} Denote $\Delta_kg(z,k):=g(z,k+1)-g(z,k)$. The correctness of the
produced recurrence follows from the computed proof certificate
\begin{equation}\label{Equ:CreaEqu}
\Delta_k g(n,k)=c_0 h(n,k)f(k)+c_1 h(n+1,k)f(k)+c_2h(n+2,k)f(k)
\end{equation}
given by $c_0=2z$, $c_1=-(2n+3)$, $c_2=2z$ and
$$g(n,k)=\frac{z^{n+1}
n!}{(2k+3)(n+k+2)!(n-k+2)!}\big[g_0\textrm{f}(k)
+g_1\textrm{f}(k+1)+g_2\textrm{f}(k+2)\big]$$ with
\begin{align*}
g_0=&8k^5-8(n-1)k^4-(z^2+28n+30)k^3+2(2n^2+(2z^2-9)n+2z^2-19)\\
&k^2+((z^2+8)n^2+(8z^2+15)n+8z^2+1)k+(n^2+3n+2)(2z^2+3)\\
g_1=&(2k+3)(k-n-2)(2k^3+(3-2n)k^2-(5n+2)k+(n+1)(z^2-3)),\\
g_2=&-(k+1)(k-n-2)(k-n-1)z^2.
\end{align*}
Namely, one can show that~\eqref{Equ:CreaEqu} holds for all
$n\geq0$ and $0\leq k\leq n$ as follows. Express $\Delta_k g(n,k)$
in terms of $f(k)$ and $f(k+1)$ by using the recurrence given
in~\myOut{\ref{MMA:recJY}} and rewrite any factorial
in~\eqref{Equ:CreaEqu} in terms of $(n+k+2)!$ and $(n-k+2)!$.
Afterwards verify~\eqref{Equ:CreaEqu} by polynomial arithmetic.
The summation of~\eqref{Equ:CreaEqu} over $k$ from $0$ to $n$
gives the recurrence in~\myOut{\ref{MMA:RecForjy}}; here we needed
the first evaluations of $f(i)=j_i(z)y_i(z)$, $i=0,1,2$,
from~(10.1.11) and~(10.1.12).
%\begin{align*}
%f(0)&=j_0(z)y_0(z)=\frac{\cos z\sin z}{z^2},\\
%f(1)&=j_1(z)y_1(z)=\frac{(z\cos z-\sin z)(\cos z+z\sin z))}{z^4}\\
%f(2)&=j_2(z)y_2(z)=-\frac{((z^2-3)\cos z-3z\sin z)((z^2-3)\sin
%z+3z\cos z)}{z^6}. \end{align*}

\smallskip

\noindent We remark that the underlying algorithms~\cite{Schneider:05} unify the creative telescoping paradigm~\cite{Zeilberger91} in the difference field setting~\cite{Schneider:07} and holonomic setting~\cite{Chyzak:00}. This general point of view opens up interesting applications, e.g., in the field of combinatorics~\cite{Schneider:05b} and particle physics~\cite{ABRS:12}.

\subsec{LODE and initial conditions for (10.1.48)}
%%%%%%%%%%%%%%%%%%%%%%%%%%%%%%%%%%%%%%%%%%%%%%%%%%%%%%%%%%%%%%%%%%%%%%%%%%%%%%%%%%

For the proof of~(10.1.48) we choose the package \textsf{\small HolonomicFunctions}.
As we have seen before, holonomic closure properties include algebraic
substitution; but since $\sin(\!\theta\!)$ is not algebraic, we have to transform
identity (10.1.48) slightly in order to make it accessible to our software:
just replace $\cos(\!\theta\!)$ by $c$ and $\sin(\!\theta\!)$ by $\sqrt{1-c^2}$.
Now it is an easy task to compute a LODE in~$z$ for the left hand side:
\begin{mma}\MLabel{MMA:lhs48}
\In |Annihilator|\left[|BesselJ|\left[0, z \sqrt{1-c^2}\right], \> |Der|[z]\right]\\
\Out \big\{z D_{\!z}^2+D_{\!z}+(z-c^2z)\big\}\\
\end{mma}

The sum on the right hand side requires some more work. Similar to identity
(10.1.49) above, the technique of creative telescoping~\cite{Zeilberger91} is
applied and it fits perfectly to the {\sf\small HolonomicFunctions}
package. The latter can deal with multivariate holonomic functions and
sequences, i.e., roughly speaking, mathematical objects that satisfy (for each
variable in question) either a LODE or a LORE of arbitrary (but fixed) order. For example,
the expression
\[
  f(n,z,c) = (4n+1)\frac{(2n)!}{2^{2n}n!^2}j_{2n}(z)P_{2n}(c)
\]
satisfies a LORE in~$n$ of order~$4$ and LODEs w.r.t. $z$ and $c$, both of
order~$2$. To derive a LODE in~$z$ for the sum we employ the following command
(the shift operator $S_{\!n}$, defined by $S_{n}f(n)=f(n+1)$, is input as
{\sf\small S}$[n]$, and the derivation $D_{\!z}$, defined by $D_{\!z}f(z)=f'(z)$, is
input as {\sf\small Der}$[z]$):
\begin{mma}
\In |CreativeTelescoping|[(4n+1)(2n)! / (2^{2n}n!^2) \, |SphericalBesselJ|[2n, z] \, |LegendreP|[2n, c],\linebreak
    |S|[n] - 1, \> |Der|[z]]\\
\Out \bigg\{\big\{z D_{\!z}^2+D_{\!z}+(z-c^2z)\big\},
  \bigg\{\frac{4(n+1)^2}{4n+5}S_{\!n}D_{\!z} + \frac{4(n+1)^2(8n^2+18n-z^2+9)}{(4n+3)(4n+5)z}S_{\!n} +
  \frac{4n^2}{4n+1}D_{\!z} + \linebreak
  \displaystyle \frac{-16c^2n^2z^2-16c^2nz^2-3c^2z^2+32n^4+40n^3+4n^2z^2+12n^2+4nz^2+z^2}{(4n+1)(4n+3)z}\bigg\}\bigg\}\\
\end{mma}
\noindent The output consists of two operators, say $P$ and~$Q$, which are called
\emph{telescoper} and \emph{certificate} (note already that $P$ equals
\myOut{\ref{MMA:lhs48}}). They satisfy the relation
\begin{equation}\label{eq.ct}
  \big(P+(S_{\!n}-1)Q\big) \, f(n,z,c) = 0,
\end{equation}
a fact that can be verified using the well-known LODEs and LOREs for spherical
Bessel functions and Legendre polynomials. Summing \eqref{eq.ct} w.r.t.~$n$
and telescoping yields
\[
  P \sum_{n=0}^{\infty} f(n,z,c) - (Qf)(0,z,c) + \lim_{n\to\infty} (Qf)(n,z,c) = 0
\]
($P$ is free of $n$ and $S_{\!n}$ and therefore can be interchanged with the
summation quantifier). Using (9.3.1) and (10.1.1) it can be shown that the
limit is~$0$, and also the part $(Qf)(0,z,c)$ vanishes.

Consequently, we have established that both sides of (10.1.48) satisfy the
same second-order LODE. It suffices to compare the initial conditions at $z=0$
(see Section~\ref{se:anal}).  For the left hand side we have $J_0(0)=1$.  From
(10.1.25) it follows that the Taylor expansion of $j_{2n}(z)$ starts with
$z^{2n}$ and hence for $z=0$ all summands are zero except the first one. With
(10.1.11) we see that the initial conditions on both sides agree.

Before turning to the next identity, we want to point to~\cite{Schneider:07}
where a different computer algebra proof of (10.1.48) has been given. More
examples of proving special function identities with the {\sf\small HolonomicFunctions}
package are collected in~\cite{KoutschanMoll11}.

%%%%%%%%%%%%%%%%%%%%%%%%%%%%%%%%%%%%%%%%%%%%%%%%%%%%%%%%%%%%%%%%%%%%%%%%%%%%%%%%%%
\subsec{LODE and initial conditions for (10.1.52)}
%%%%%%%%%%%%%%%%%%%%%%%%%%%%%%%%%%%%%%%%%%%%%%%%%%%%%%%%%%%%%%%%%%%%%%%%%%%%%%%%%%
Again we compute a LODE with {\sf\small Sigma}.  In order to get a LODE of the left
hand side of~(10.1.52) we compute a LODE of its truncated version
\begin{mma}
\In |mySum|=\sum_{k=0}^a j[k]^2;\\
\end{mma}
\noindent Note that the summand of our input-sum depends
non-linearly on $j_k(z)$. In order to handle this type of
summation input, {\sf\small Sigma} needs in addition the
package~\cite{Gerhold:02}
\begin{mma}
\In << |OreSys.m| \\
\Print OreSys package by Stefan Gerhold
\copyright\ RISC-Linz\\
\end{mma}
\noindent for uncoupling systems of LODE-systems. Then using a new
feature of {\sf\small Sigma} we can continue as ``as usual''.  Given the
difference-differential equation of the form~\eqref{recZ} for
$j(k):=j_k(z)$ and $j^{(0,1)}(k,z):=\frac{\d}{\d z}j_k(z)$:
\begin{mma}\MLabel{MMA:recZ}
\In |recZ|=j^{(0,1)}[k,z]==\frac{k}{z}j[k]+j[k+1];\\
\end{mma}
\noindent see~(10.2.20), and the recurrence~\myIn{\ref{MMA:recJ}}
of the form~\eqref{recN}, we compute a LODE for
$\textrm{mySum}(=\textrm{SUM}[n])$:
\begin{mma}
\In |mySum|=\sum_{k=0}^a j[k]^2;\\
\end{mma}
\begin{mma}\MLabel{recSumJ2}
\In |GenerateDE|[|mySum|,n,\{|recJ|,j[k]\},|recZ|]\\
\Out z|SUM|'[z]+|SUM|[z]== (z j[a]j[a+1]-(2a+1)j[a]^2) - (z j[0]j[1]-j[0]^2) \\
\end{mma}

\noindent\textit{A correctness proof.} The correctness of the LODE can be checked by the computed proof
certificate
\begin{equation}\label{Equ:CreaD}
\Delta_kg(z,k)=c_0 j(k)^2+c_1 j^{(0,1)}(k,z)^2
\end{equation}
with $c_0=1$, $c_1=z$ and $g(z,k)=zj(k)j(k+1)-(2k+1)j(k)^2$.
Namely, one can easily show that~\eqref{Equ:CreaD} holds for all
$0\leq k$ as follows. Express~\eqref{Equ:CreaD} in terms of $j(k)$
and $j(k+1)$ by using the recurrence given
in~\myIn{\ref{MMA:recJ}} and the difference-differential equation
given in~\myIn{\ref{MMA:recZ}}. Afterwards
verify~\eqref{Equ:CreaD} by polynomial arithmetic. Then
summing~\eqref{Equ:CreaD} over $k$ from $0$ to $a$ gives the
recurrence in~\myOut{\ref{recSumJ2}}; here we used the initial
values~(10.1.11).

\medskip

\noindent Next, we let $a\to\infty$. Then $j_a(z)$ tends to zero by~(9.3.1).
Therefore, the left hand side of~(10.1.52) satisfies the LODE
  \begin{equation}\label{odeJS}
    S(z) + z\frac{\mathrm{d}S(z)}{\mathrm{d}z} = \frac{\sin(2z)}{2z}.
  \end{equation}
It is readily checked that the right hand side satisfies it, too,
and both sides equal~$1$ at $z=0$. This establishes equality of
both sides of~(10.1.52).

Alternatively, we can derive the inhomogeneous differential equation
for the left hand side of (10.1.52) with {\sf\small HolonomicFunctions}:
\begin{mma}
\In |Annihilator|[|Sum|[|SphericalBesselJ|[n,z]^2, \{n, 0, |Infinity|\}], \> |Der|[z], \> |Inhomogeneous| \to |True|]\\
\Out \{\{zD_{\!z}+1\}, \{\text{Hold}[\text{Limit}[\dots, n\to\infty]] + \dots\}\}\\
\end{mma}
\noindent The output consists of a differential operator and an expression that gives
the inhomogeneous part (abbreviated above). Without help, Mathematica is not
able to simplify the latter (i.e., compute the limit), but using (9.3.1) it
succeeds and we get
\[
  \big(zD_{\!z}+1\big)S(z) - \frac{\sin(z)\cos(z)}{z} = 0
\]
which of course agrees with~\eqref{odeJS}.

\section{Series Solutions of LODEs and Analyticity}\label{se:anal}

In some proofs we have determined a differential equation that is satisfied by both
sides of the identity in question, and then compared initial values.
In contrast to the case of recurrences, the validity of this approach needs some non-trivial justification.
This procedure can be justified by well-known uniqueness results for solutions of LODEs,
to be outlined in this section.
In the proofs of (10.1.39), (10.1.40), (10.2.30), and (10.2.31), the point $t=0$, where we checked
initial conditions, is an ordinary point of the LODE (i.e., the leading coefficient of the
LODE does not vanish at $t=0$).
Then there is a unique analytic solution, if the number of prescribed initial values equals
the order of the equation.
The identity then holds (at least) in the domain (containing zero) where we can establish analyticity of both sides.

\begin{proposition}
  Identity (10.1.40) holds for all complex $t$ and all complex $z\neq 0$.
  The same is true for~(10.2.31).
\end{proposition}
\begin{proof}
  We consider~(10.1.40) and omit the analogous considerations for~(10.2.31).
  For $n\in\mathbb{Z}$, the function $j_{n-1}(z)$ is defined for $z\in\set{C}^*$.
  We fix such a $z$ and consider both sides of~(10.1.40) as
  functions of $t$.
  By~(9.3.1), the right hand side
  converges uniformly for all complex $t$, therefore it is an entire function of $t$.
  The left hand side is also entire, since $\cos\sqrt{w}=\sum_{n\geq 0}(-1)^nw^n/(2n)!$
  is an entire function of $w$.
  Initial values at $t=0$ and an LODE satisfied by both sides were already presented in Section~\ref{se:comp}, hence,
  by the above uniqueness property, identity~(10.1.40) is proved.\qed
\end{proof}

\begin{proposition}
  Identity (10.1.39) holds for all complex $z$ and $t$ with $|\Im(z)|\leq \Re(z)$
  and $2|t|<|z|$. If $|\Im(z)|\leq -\Re(z)$,
  then the identity holds with switched sign for all $t$ with $2|t|<|z|$.
  The same is true for~(10.2.30).
\end{proposition}
\begin{proof}
  We give the proof in the case of~(10.1.39); (10.2.30) is treated analogously.
  First we complete the check of initial values from Section~\ref{se:comp}.
  For $t=0$, the right hand side is $y_{-1}(z) = (\sin z)/z$, and on the left hand
  side we have $(\sin\sqrt{z^2})/z$. Thus, at $t=0$ both sides agree for $|\arg(z)|<\pi/2$,
  which follows from $|\Im(z)|\leq \Re(z)$;
  for $\pi/2<|\arg(z)|<\pi$, which follows from $|\Im(z)|\leq -\Re(z)$,
  the identity holds at $t=0$ with switched sign, because the function
  $w\mapsto \sqrt{w^2}$ changes sign when crossing the branch cut $\i\set{R}$.
  The first derivatives at $t=0$ are $(\cos\sqrt{z^2})/\sqrt{z^2} = (\cos z)/\sqrt{z^2}$
  and $-y_0(z)=(\cos z)/z$,
  respectively. The same consideration as for the first initial value completes the check of the initial conditions.

  Now we show that both sides of~(10.1.39) are analytic functions of $t$ for fixed $z\neq 0$
  with $|\Im(z)|\leq |\Re(z)|$.
  Let us start by determining the radius of convergence of the right hand side. 
	It is an easy consequence of (9.3.1) that

	\[
  y_n(z) \sim -\frac{\sqrt{2}}{z} \left( \frac{2n}{\mathrm{e}z}\right)^n,
    \quad n\to\infty,\ z\neq 0.
\] 	
Hence, by Stirling's formula, the radius of convergence is $|z|/2$, and
so the right hand side is analytic for $2|t|<|z|$.
	
	%It is an easy
  %consequence of~(9.3.2) that
  %\[
  %  y_n(z) = \mathrm{O}\left( \left(\frac{2n}{\e|z|} \right)^n\right), \qquad n\to\infty, \quad z\neq 0.
  %\]
  %Hence, by Stirling's formula, the summand is $\mathrm{O}(n^{-3/2}(2|t|/|z|)^n)$.
  %This shows that the radius of convergence
  %is $|z|/2$. Hence the right hand side is analytic for $2|t|<|z|$.

  The left hand side of~(10.1.39) has a branch cut along a half-line starting
  at $t=-z/2$, a point on the circle of convergence of the right hand side.
  If this half line has no other intersection with this circle,
  then the left hand side is analytic in the disk $\{t:2|t|<|z|\}$. Otherwise,
  the branch cut separates the disk into two segments, and the identity does not necessarily hold
  in a segment that does not contain $t=0$. As we will now show, our assumptions exclude
  the possibility of a second intersection.
  Once again it is convenient to proceed by computer algebra. Note that the presence
  of two intersections means that
  \[
    \Bigl(\exists s\neq t\in\set{C}\Bigr)
    \Bigl(2|s|=|z| \wedge 2|t|=|z| \wedge z^2+2zs \in {]{-\infty},0]} \wedge z^2+2zt \in {]{-\infty},0]}\Bigr)
  \]
  holds. Upon rewriting this formula with real variables,
  it can be simplified by Mathematica's Reduce command; the result
  -- translated back into complex language -- is the equivalent formula
  $|\Re(z)|<|\Im(z)|$. Summing up, under our assumptions on $z$ the left hand side
  of~(10.1.39) is analytic in the disk $\{t:2|t|<|z|\}$.\qed
\end{proof}

%
%Remark: The treatment of~(10.1.39) and (10.2.30) for $|\Im(z)|> \Re(z)$ poses no problems.
%

\medskip

Now consider the LODE~\eqref{odeJS}, which we want to employ to
prove (10.1.52). The point $z=0$ is not an ordinary point, so the
question of uniqueness of the solution is more subtle. The origin
is a regular singular point of~\eqref{odeJS}, since the degree of
the indicial polynomial
\[
  [z^0]p_s(z)^{-1} z^{s-\sigma} \mathcal{L} z^\sigma = \sigma + 1
\]
agrees with the order $s=1$ of the LODE. Here,
$p_s(z)=z$ denotes the leading coefficient, and $\mathcal{L}$ the differential operator
\[
  \mathcal{L} := 1 + z D_{\!z}.
\]

The following classical result~\cite{Ince26} describes the structure of a fundamental system
at a regular singular point. See also the concise exposition in Meunier and Salvy~\cite{MeSa03}.

\begin{theorem}\label{thm:ode}
  Let $z=0$ be a regular singular point of a homogeneous LODE of order~$s$. Denote
  the roots of the indicial polynomial by $\sigma_1,\dots,\sigma_s$, and let $m_1,\dots,m_s$ be their multiplicities.
  Then the equation has a basis of $s$ solutions
  \begin{equation}\label{eq:log series}
    z^{\sigma_i} \sum_{j=0}^{d_i} \log^j(z) \Phi_{ij}(z),\qquad 1\leq i\leq s,
  \end{equation}
  where $d_i<s$, and the $\Phi_{ij}(z)$ are convergent power series. Each of these solutions
  is uniquely defined by the coefficients of the $s$ ``monomials''
  \[
    \bigcup_{i=1}^s \left\{ z^{\sigma_i}, z^{\sigma_i}\log z,\dots, z^{\sigma_i}\log^{m_i-1} z  \right\}
  \]
  in the series~\eqref{eq:log series}.
\end{theorem}

\begin{proposition}
  Identity~(10.1.52) holds for all $z\in\set{C}$.
\end{proposition}
\begin{proof}
  We have shown in the preceeding section that both sides satisfy the LODE \eqref{odeJS}.
  As seen above, the indicial polynomial of the homogeneous equation ${\mathcal L}f=f + zf'=0$ is $\sigma+1$.
  Hence, by Theorem~\ref{thm:ode}, a solution of ${\mathcal L}f = 0$
  that has the form~\eqref{eq:log series} is uniquely defined by the coefficient of $z^{-1}$.
  Hence the zero function is the only analytic solution of the homogeneous initial value problem
  ${\mathcal L}f=0$, $f(0)=0$. It is a trivial consequence that the inhomogeneous equation~\eqref{odeJS}
  cannot have more than one analytic solution with $f(0)=1$.
  Therefore, (10.1.52) holds in a neighbourhood of $z=0$.
  The left hand side of (10.1.52) is entire since it is a uniform limit of entire functions,
  and the right hand side is entire by (5.2.14).
  Thus, the identity holds in the whole complex plane by analytic continuation.\qed
\end{proof}

\begin{proposition}
  Identity (10.1.48) holds for all complex $z$ and $\theta$.
\end{proposition}
\begin{proof}
  By the Laplace-Heine formula~\cite[Theorem~8.21.1]{Sz75}, $P_{2n}(\cos\theta)$ grows at most exponentially as $n\to\infty$.
  Together with~(9.3.1) and $n(2n)!/(2^{2n}n!^2)=\mathrm{O}(\sqrt{n})$, this shows that the right hand side of~(10.1.48) is an
  entire function of $z$ and $\theta$.
  In Section~\ref{Sec:SymSum} we showed that both sides of (10.1.48) satisfy
  the differential equation $zf''(z)+f'(z)+z(1-c^2)f(z)=0$ (whose indicial
  equation is $\sigma^2=0$) and that the initial condition at $z^0$ agrees.
  The result follows from Theorem~\ref{thm:ode} and the fact that both sides
  are entire functions.\qed
\end{proof}

\section{Non-Computer Proofs}

Some of our identities can be easily proved from some of the others,
without using any software machinery.
The computer proofs that we have in hand suffice for
establishing the remaining identities (10.1.42), (10.1.43), (10.1.44), and (10.2.34) in this spirit.
The reader should by now be convinced that, if desired, all of them can also be
proved by the algorithmic methods we have presented.

\begin{proposition}
  Identities (10.1.42), (10.1.43), and (10.1.44) follow from (10.1.41).
  They hold for $z\in\set{C}\setminus\set{R}_{\leq 0}$.
\end{proposition}
\begin{proof}
Identities (10.1.42), (10.1.43), and (10.1.44) can be done analogously to (10.1.41), but we instead
present (non-computer) deductions from (10.1.41).
The derivative of $Y_\nu$ w.r.t. $\nu$ can be expressed in terms of $J_\nu$,
$J_{-\nu}$, and $Y_\nu$, see (9.1.65) in the appendix.
Note that $\cot{(\nu+1/2)\pi}$ vanishes for $\nu=0,1$. (9.1.65) thus yields
\[
  \left[\der{\nu}y_\nu(z)\right]_{\nu=0}
  = \left[\der{\nu}j_\nu(z)\right]_{\nu=-1} -\frac{\pi\sin z}{z}
\]
and
\[
  \left[\der{\nu}y_\nu(z)\right]_{\nu=-1}
   = -\left[\der{\nu}j_\nu(z)\right]_{\nu=0} -\frac{\pi\cos z}{z}.
\]
Therefore, we have a relation between the left hand sides of (10.1.42) and (10.1.43),
and one between the left hand sides of (10.1.41) and (10.1.44).
It is easy to verify that the respective right hand sides satisfy
the same relations.
Hence the assertion will be established once we show that (10.1.42)
follows from (10.1.41). To this end,
it suffices to show that the left hand sides of these identities satisfy
\begin{equation}\label{eq:lhsides}
  \der{z}\left( z \left[ \der{\nu} j_\nu(z) \right]_{\nu=0} \right)
  - z \left[ \der{\nu}j_\nu(z) \right]_{\nu = -1} = -\frac{\sin z}{z},
\end{equation}
since once again it is easy to see that the right hand sides of (10.1.41)
and (10.1.42) obey the same relation.
By (9.1.64), the recurrence relations of $\Gamma$ and~$\psi$,
and the duplication formula of~$\Gamma$, the left hand side of~\eqref{eq:lhsides} equals
\begin{align*}
  \phantom{=}& \frac{\sin z}{z}-\sqrt{\pi} \sum_{k=0}^\infty(-\tfrac14)^k
   \left( \frac{\psi(k+\tfrac32)}{\Gamma(k+\tfrac32)/(k+\tfrac12)}
   -\frac{\psi(k+\tfrac12)}{\Gamma(k+\tfrac12)}  \right) \frac{z^{2k}}{k!} \\
  =& \frac{\sin z}{z}-\sqrt{\pi} \sum_{k=0}^\infty(-\tfrac14)^k \frac{1}{\Gamma(k+\tfrac12)(k+\tfrac12)}
   \frac{z^{2k}}{k!} \\
  =& \frac{\sin z}{z} - \sum_{k=0}^\infty(-\tfrac14)^k \frac{2^{2k+1}z^{2k}}{\Gamma(2k+2)}
   %= \frac{\sin z}{z}- 2\sum_{k=0}^\infty \frac{(-1)^k z^{2k}}{(2k+1)!}
   = -\frac{\sin z}{z}.
\end{align*}

\vspace*{-0.8cm}\qed
\end{proof}

\begin{proposition}
  Identity~(10.2.34) follows from~(10.1.32) and (10.2.33).
  It holds for $z\in\set{C}\setminus\set{R}_{\leq 0}$.
\end{proposition}
\begin{proof}
  Indeed, by (9.6.43) we have
  \begin{align*}
    \left[\frac{\partial}{\partial \nu}K_\nu(z)\right]_{\nu=\pm 1/2} &=
    \frac{\pi}{2}\csc(\nu \pi)\left[ \frac{\partial}{\partial \nu}I_{-\nu}(z)
     - \frac{\partial}{\partial \nu}I_\nu(z)\right]_{\nu=\pm 1/2} \\
     &= - \frac{\pi}{2}\csc(\nu \pi)\left(\left[\frac{\partial}{\partial \nu}I_\nu(z)\right]_{\nu=\mp 1/2}
     + \left[\frac{\partial}{\partial \nu}I_\nu(z)\right]_{\nu=\pm 1/2} \right) \\
     &= \pm \sqrt{\frac{\pi}{2 z}} \mathrm{e}^z \mathrm{E}_1(2z).
   \end{align*}

\vspace*{-0.8cm}\qed
\end{proof}

Finally, we note that~(10.2.32), which was proved in Proposition~\ref{pr:10.2.32},
can be proved by hand from~(10.1.41).
Indeed, replacing $z$ with $\i z$ in (9.1.64) makes the $k$-sum
in (9.1.64) equal the $k$-sum in (9.6.42). Solving both relations for
the $k$-sum allows to express $\der{\nu}I_\nu(z)$ by $I_\nu(z)$,
$J_\nu(\i z)$, and $\der{\nu}J_\nu(\i z)$. Plugging in $\nu=\tfrac12$, rewriting $\der{\nu}J_\nu(\i z)$
with~(10.1.41), and using the relations~(5.2.21) and (5.2.23)
between the exponential integral and the sine and cosine integrals
gives~(10.2.32).
Analogously, (10.2.33) follows from~(10.1.42).

%\begin{proposition}
%  Identity (10.2.32) follows from (10.1.41).
%  It holds for $z\in\set{C}\setminus\set{R}_{\leq 0}$.
%\end{proposition}
%Remark: (10.2.33) can be proved analogously from (10.1.42). Maybe we just sketch the proofs,
%since Burki has proofs with Zeilberger's algo.
%\begin{proof}
%  Replacing $z$ with $\i z$ in (9.1.64) makes the $k$-sum
%  in (9.1.64) equal the $k$-sum in (9.6.42).
%  This observation leads to
%  \begin{equation}\label{eq:diff I}
%    \i^{-\nu}\left( J_\nu(\i z)\log\frac{\i z}{2}
%     - \der{\nu}J_\nu(\i z) \right) =
%     I_\nu(z)\log\frac{z}{2} - \der{\nu}I_\nu(z), \qquad \Im z<0.
%  \end{equation}
%  For $\nu=\tfrac12$ we can express $\partial J_\nu(\i z)/\partial\nu$
%  via (10.1.41), und using (5.2.23) and (5.2.21)
%  we thus find
%  \begin{align*}
%    \left[\der{\nu}I_\nu(z)\right]_{\nu=1/2}
%     &= \frac{1}{\sqrt{2\pi z}}\left( \i\pi \e^{-z}
%     + \e^{-z} \E1(-2z) - \e^{z}\E1(2z)\right) \\
%    &= -\frac{1}{\sqrt{2\pi z}}\left( \e^{-z} \Ei (2z)
%     + \e^{z}\E1(2z)\right), \qquad \Im{z}<0.
%  \end{align*}
%  The last equality follows from the identity
%  \[
%    \E1(-z) = -\Ei (z) - \i\pi,\qquad \Im z < 0,
%  \]
%  which follows from the series expansions (5.1.10) and (5.1.11).
%  Finally, the constraint $\Im z < 0$ can be removed by analytic continuation.
%\end{proof}

\begin{acknowledgement}
This work has been supported by the Austrian Science Fund (FWF) grants P20347-N18, P24880-N25, Y464-N18, DK W1214 (DK6, DK13)
 and SFB F50 (F5004-N15, F5006-N15, F5009-N15), and
by the EU Network LHCPhenoNet PITN-GA-2010-264564.
\end{acknowledgement}

\section*{Appendix: List of Relevant Table Entries}

For the reader's convenience, we collect here all identities from Abramowitz, Stegun~\cite{AbSt73}
that we have used.

\allowdisplaybreaks
\begin{alignat}3
 \tag{5.1.10}
 \Ei(x)&=\gamma+\ln x+\sum_{n=1}^\infty\frac{x^n}{n\,n!}
 &&(x>0)
 \\
 \tag{5.1.11}
 \E1(z)&=-\gamma-\ln z-\sum_{n=1}^\infty\frac{(-1)^nz^n}{n\,n!}
 &&\hspace*{-0.5cm}(|\arg z|<\pi)
 \\
 \tag{5.2.14}
 \Si(x)&=\sum_{n=0}^\infty\frac{(-1)^nx^{2n+1}}{(2n+1)(2n+1)!}
 \\
 \tag{5.2.16}
 \Ci(x)&=\gamma+\log x+\sum_{n=1}^\infty\frac{(-1)^nx^{2n}}{2n(2n)!}
 \\
 \tag{5.2.21}
 \Si(x)&=\frac1{2\i}(\E1(\i z)-\E1(-\i z))+\frac\pi2
 &&\hspace*{-0.5cm}(|\arg z|<\frac\pi2)
 \\
 \tag{5.2.23}
 \Ci(x)&=-\frac12(\E1(\i z)+\E1(-\i z))
 &&\hspace*{-0.5cm}(|\arg z|<\frac\pi2)
 \\
 \tag{6.3.5}
 \psi(z+1)&=\psi(z)+\frac1z
 \\
 \tag{9.1.64}
 \der\nu
 J_\nu(z)&=J_\nu(z)\log(\tfrac12z)-(\tfrac12z)^\nu\sum_{k=0}^\infty(-1)^k\frac{\psi(\nu+k+1)}{\Gamma(\nu+k+1)}\frac{(\tfrac14z^2)^k}{k!}
 \kern-70pt\null
 \\
 \tag{9.1.65}
 \der\nu\ Y_\nu(z)&=\cot(\nu\pi)\Bigl(\der\nu J_\nu(z)-\pi
 Y_\nu(z)\Bigr)\\
 \notag
 &\qquad{}
  -\csc(\nu\pi)\der\nu J_{-\nu}(z)-\pi J_\nu(z)
 &&(\nu\neq0,\pm1,\pm2,\dots)
 \\
 \tag{9.3.1}
 J_\nu(z)&\sim\frac1{\sqrt{2\pi\nu}}\Bigl(\frac{\e z}{2\nu}\Bigr)^{\!\nu},\quad
 Y_\nu(z)\sim-\sqrt{\frac2{\pi\nu}}\Bigl(\frac{\e z}{2\nu}\Bigr)^{\!-\nu}\>\>
 && (\nu\to\infty) % \kern-30pt\null
 \\
% \tag{9.3.2}
% J_\nu(\nu\sech\alpha)&\sim\frac{\e^{\nu(\tanh\alpha-\alpha)}}{\sqrt{2\pi\nu\tanh\alpha}}
% &&(\alpha>0)\\\notag
% %Y_\nu(\nu\sech\alpha)&\sim-\frac{\e^{\nu(\alpha-\tanh\alpha)}}{\sqrt{\tfrac12\pi\nu\tanh\alpha}}
% &&(\alpha>0)
% \\
 \tag{9.6.10}
 I_\nu(z)&=(\tfrac12z)^\nu\sum_{k=0}^\infty\frac{(\tfrac14z^2)^k}{k!\Gamma(\nu+k+1)}
 \\
 \tag{9.6.42}
 \der\nu I_\nu(z)&=I_\nu(z)\ln(\tfrac12z)-(\tfrac12z)^\nu\sum_{k=0}^\infty\frac{\psi(\nu+k+1)}{\Gamma(\nu+k+1)}\frac{(\tfrac14z^2)^k}{k!}
 \kern-50pt\null
 \\
 \tag{10.1.1}
 j_n(z)&=\sqrt{\tfrac12\pi/z}J_{n+\frac12}(z),
 \quad
 y_n(z)=\sqrt{\tfrac12\pi/z}Y_{n+\frac12}(z)
 \kern-50pt\null
 \\
 \tag{10.1.11}
 j_0(z)&=\frac{\sin z}z,
 \quad
 j_1(z)=\frac{\sin z}{z^2}-\frac{\cos z}z,\\
 \notag
 j_2(z)&=\Bigl(\frac3{z^3}-\frac1z\Bigr)\sin z-\frac3{z^2}\cos z
 \\
 \tag{10.1.12}
 y_0(z)&=-j_{-1}(z)=-\frac{\cos z}z,
 \quad
 y_1(z)=j_{-2}(z)=-\frac{\cos z}{z^2},\kern-50pt\null\\
 \notag
 y_2(z)&=-j_{-3}(z)=\Bigl(\frac1z-\frac3{z^2}\Bigr)\cos
 z-\frac3{z^2}\sin z\\
 \tag{10.1.19}
 j_{n-1}(z)&+j_{n+1}(z)=(2n+1)z^{-1}j_n(z)
 &&(n\in\set Z)
 \\\notag
 y_{n-1}(z)&+y_{n+1}(z)=(2n+1)z^{-1}y_n(z)
 &&(n\in\set Z)
 \\
 \tag{10.1.25}
 j_n(z)&=z^n\Bigl(-\frac1z\der z\Bigr)^{\!n}\,\frac{\sin z}z
 \\
 \tag{10.2.20}
 \frac{n+1}{z}&j_n(z)+\frac{\d}{\d z}j_n(z)=j_{n-1}(z)
\end{alignat}

\bibliographystyle{alpha}
%\bibliography{../gerhold}

\end{document}